\documentclass{article}
\usepackage{float}
\usepackage{lmodern}

\PassOptionsToPackage{dvipsnames}{xcolor}
\usepackage{xcolor}

\usepackage{arxiv}

\usepackage[utf8]{inputenc}   
\usepackage[T1]{fontenc}
\usepackage{lmodern}
\usepackage{microtype}

\usepackage{amsmath,amssymb,amsfonts}
\usepackage{amsthm}
\usepackage{mathtools}
\usepackage{physics}
\newtheorem{theorem}{Theorem}

\usepackage{booktabs}
\usepackage{nicefrac}

\usepackage{graphicx}
\graphicspath{{./images/}}
\usepackage{tikz}
\usetikzlibrary{calc}

\usepackage{algorithm}
\usepackage[noend]{algpseudocode}

\usepackage{url}
\usepackage{lipsum}

\usepackage{cleveref}

\graphicspath{ {./images/} }

\title{Navigating Quantum Missteps in Agent-Based Modeling: A Schelling Model Case Study}

\author{
     C. Nico Barati \\
   Center for Secure and Intelligent Critical Systems \\
  Old Dominion University \\
  Norfolk, VA 23259 \\
   \texttt{cbaratin@odu.edu} \\
   \And
  Arie Croitoru \\
  Department of Computational and Data Science \\
  George Mason University \\
 Fairfax, VA 22030 \\
   \texttt{acroitor@gmu.edu} \\
    \And
   Ross Gore \\
   Center for Secure and Intelligent Critical Systems\\
    Office of Enterprise Research and Innovation\\
  Old Dominion University \\
  Norfolk, VA 23259 \\
   \texttt{rgore@odu.edu} \\
   \And
 Michael Jarret \\
  Quantum Science and Engineering Center \\
  Center for Social Complexity \\
  Department of Mathematical Sciences \\
  Department of Computer Science \\
  George Mason University \\
  Fairfax, VA 22030 \\
  \texttt{mjarretb@gmu.edu} \\
   \And
 William G. Kennedy \\
  Center for Social Complexity \\
  Department of Computational and Data Science \\
  George Mason University\\
  Fairfax, VA 22030 \\
  \texttt{wkennedy@gmu.edu} \\
   \AND
   Andrew Maciejunes \\
    Office of Enterprise Research and Innovation\\
  Old Dominion University \\
  Norfolk, VA 23259 \\
   \texttt{amacieju@odu.edu} \\
     \And
 Maxim A Malikov \\
  Department of Communication\\
  University of California Davis\\
 Davis, CA 95616 \\
  \texttt{mmalikov@ucdavis.edu} \\
   \And
   Samuel S. Mendelson \\
   Quantum Science and Engineering Center \\
   Department of Mathematical Sciences \\
   Department of Computer Science \\
   George Mason University
   Fairfax, VA 22030 \\
   \texttt{smendels@gmu.edu} \\  
}

\begin{document}
\maketitle
\begin{abstract}
Quantum computing promises transformative advances, but remains constrained by recurring misconceptions and methodological pitfalls. This paper demonstrates a fundamental incompatibility between traditional agent-based modeling (ABM) implementations and quantum optimization frameworks like Quadratic Unconstrained Binary Optimization (QUBO). Using Schelling's segregation model as a case study, we show that the standard practice of directly translating ABM state observations into QUBO formulations not only fails to deliver quantum advantage, but actively undermines computational efficiency. The fundamental issue is architectural. Traditional ABM implementations entail observing the state of the system at each iteration, systematically destroying the quantum superposition required for computational advantage. Through analysis of Schelling's segregation dynamics on lollipop networks, we demonstrate how abandoning the QUBO reduction paradigm and instead reconceptualizing the research question, from "simulate agent dynamics iteratively until convergence" to "compute minimum of agent moves required for global satisfaction", enables a faster classical solution. This structural reconceptualization yields an algorithm that exploits network symmetries obscured in traditional ABM simulations and QUBO formulations. It establishes a new lower bound which quantum approaches must outperform to achieve advantage. Our work emphasizes that progress in quantum agent-based modeling does not require forcing classical ABM implementations into quantum frameworks. Instead, it should focus on clarifying when quantum advantage is structurally possible, developing best-in-class classical baselines through problem analysis, and fundamentally reformulating research questions rather than preserving classical iterative state change observation paradigms.
\end{abstract}



\section{Introduction}
Quantum computing has witnessed rapid growth over the past decade, accompanied by significant investment and widespread claims of impending technological revolution. However, this acceleration has also revealed a recurring pattern of conceptual and methodological missteps. Many studies: (1) proclaim quantum advantage without establishing clear classical baselines, (2) conflate heuristic and adiabatic models, or (3) apply quantum tools to problems whose underlying structure cannot meaningfully benefit from quantum effects. As a result, there is a hype problem in which the theoretical potential is conflated with practical demonstration and where terminological imprecision leads to misinterpretation of results \cite{quantum_myth_busters_2025,quantum_hype_reality_2025,dont_believe_hype_2025}. This misrepresentation, often unintentional, may destroy the credibility of quantum technologies as it becomes increasingly apparent that they cannot meet inaccurate and unfair expectations \cite{google_dwave_speedup_2016,quantum_winter_warning_2025,experts_weigh_microsoft_2025}.

These tendencies are visible in efforts to adapt agent-based modeling (ABM) to quantum computing \cite{cmu_qubo_2025,penalty_free_qaoa_2024}. ABMs, which simulate complex systems through the local interactions of autonomous agents, are richly expressive but computationally demanding. The promise of quantum acceleration has attracted substantial interest. However, researchers assume that their existing frameworks can be adapted to quantum computing simply by retooling them \cite{qubo_criticism_2024,qubo_constraints_matlab_2024,qubit_efficient_qubo_2023}. This overlooks the deeper structural differences between classical bottom-up modeling implementations and quantum paradigms. Resulting research performed in this manner reflects superficial adaptations that obscure true opportunities for quantum innovation. These attempts to \emph{quantize} ABMs typically yield slower, less interpretable, and potentially inaccurate implementations. In this paper, we analyze the missteps ABM researchers take when attempting to apply quantum algorithms in this manner. Then, we provide guidance for improving the quality of quantum algorithm research. 

To demonstrate how this guidance can be applied in practice, we present a case study of our attempts to develop an effective quantum algorithm related to Thomas Schelling's model of segregation \cite{schelling1971dynamic}. Schelling's model of residential segregation  remains one of the most well-known and influential ABMs. It demonstrates how even mild individual preferences can lead to highly segregated patterns. In the model, agents of two types, representing different demographic groups, occupy cells on a grid or a network, and relocate when their local neighborhood fails to meet their preferred proportion of similar agents. Figure~\ref{fig:schelling_start_and_end}A shows the initial state of the model. Figure~\ref{fig:schelling_start_and_end}B shows an example of the final segregated state in which the agents are satisfied. Despite its simplicity, Schelling's model continues to inspire new research questions and real-world applications \cite{clark2008understanding, silver2021venues, abella2022aging, gambetta2023mobility, gunaratne2023generating, xu2024homophily}. \begin{figure}[!ht]
\centering
\includegraphics[width=0.9\columnwidth]{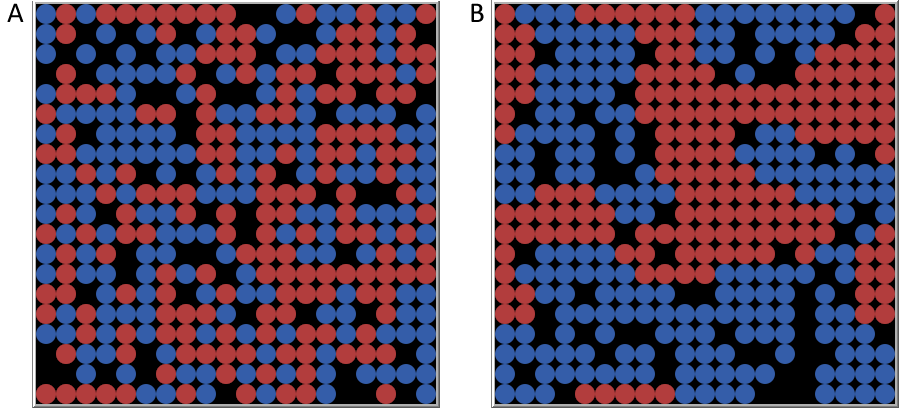}
\caption{(A) Initial state of Schelling's segregation model. (B) Final state of Schelling model with satisfied agents.}\label{fig:schelling_start_and_end}
\end{figure}


Here, we focus on developing an efficient algorithm to compute the minimum number of agent moves required to achieve global satisfaction across various network topologies in Schelling's model. This metric is important because it defines susceptibility to segregation for a model instance (e.g. network geometry, number of agents of each type, same type preference, etc). We began by constructing a Quadratic Unconstrained Binary Optimization (QUBO) formulation of Schelling's model to compute the number of agent moves required for global satisfaction. The effort fell victim to pitfalls. 
These included: (1) problem misalignment, (2) excessive encoding overhead, (3) loss of interpretability, and (4) the realization that there does not exist a quantum computer to run our solution with a non-trivial number of agents (e.g. 1,000 agents) \cite{preskill2018quantum}. 

We abandoned this approach and reformulated our research question. Instead of focusing on "simulating agent dynamics iteratively until convergence" we focused on any method to "compute the minimum number of agent moves for moves for global satisfaction". We identified a specific network topology, the lollipop network, whose structure can yield model instances with an extremely larger number of agent moves required to achieve global satisfaction.

By focusing on the structure of the network, we developed a very efficient classical algorithm to compute the required agent moves for global satisfaction. This solution sets a new lower complexity bound for this problem \cite{kreisel2022equilibria}.

Considering the overhead incurred by quantum state preparation and measurement, it is unlikely that any quantum approach will outperform this classical method with respect to wall-clock time in the foreseeable future. 

It is important to note that we did not directly address the question of whether an interesting quantum algorithm might exist for computing the required agent moves for global satisfaction. However, our  experience still illustrates how rigorous structural analysis and precise problem formulation, conducted prior to quantum algorithm design, can yield significant classical breakthroughs, while establishing concrete benchmarks for evaluating future claims of quantum advantage. Through this case study, we argue that genuine progress in improving the efficiency of ABM demands structural reconceptualization of research questions to favor formulations where quantum phenomena, or even existing classical algorithms, may offer real advantage.

\section{Quantum Missteps}
Most quantum missteps stem from a lack of clarity, consistency, or rigor in how key concepts and terminology are used. Quantum computing exploits strange properties of the universe that we know to be true, but nonetheless don't quite understand. In the words of Richard Feynman, 
\begin{quote}
    [Quantum theory] describes nature as absurd from the point of view of common sense. And yet it fully agrees with experiment. \cite{feynman1985qed}.
\end{quote}
A consequence of this observation is that applying classical intuitions and methodologies when developing quantum algorithms is not merely ineffective, it may be anti-productive.\footnote{We distinguish anti-productive from counterproductive: the former actively undermines progress, while the latter simply fails to advance it.} Indeed, quantum contextuality, a phenomenon that directly violates classical intuition, has been identified as a computational resource in certain quantum computing models \cite{Howard2014Contextuality}. Contextuality, alongside other foundational principles of quantum mechanics, underscores a critical departure from classical reasoning. In quantum computation, counterfactual outcomes—those that could have occurred but did not—carry as much, if not more, significance as observed outcomes. This feature of quantum theory fundamentally undermines the classical software development paradigm of iterative trial-and-error refinement. One cannot learn by observing the behavior of code that was never executed, yet quantum advantage may hinge precisely on such unexecuted computational paths. 

To exploit the power of quantum computers, one must think not only about how to map an existing problem and algorithm onto a quantum machine, but also whether it is even possible in the first place for such a mapping to achieve quantum advantage and through what phenomena. This requires an examination of the landscape of ongoing research, where it becomes clear that careful attention to these details is essential for both scientific progress and credible communication.

\subsection{The Current Landscape}

One of the most fundamental issues in current quantum computing research is a strong divergence in the research community. Now that quantum computers have entered and are beginning to emerge from the Noisy-Intermediate-Scale-Quantum regime \cite{preskill2018quantum}, many long-settled debates are being rehashed without the depth required to produce useful results. In Scott Aaronson's words
\begin{quote}
    In quantum computing [$\dots$] there’s right now a race for practice to catch up to where theory has been since the mid-1990s \cite{aaronson2023} \footnote{We have taken this quote out of context, but believe it still summarizes the current state of research. Scott's original intent was to note that experimental hardware cannot yet achieve what theory knows is possible. Nonetheless, we anticipate he would agree with our current use as another side of the same coin.}. 
\end{quote}

Comparing quantum computers to classical competitors requires substantially more than ``good'' results. Instead, it is necessary to determine that the machines are genuinely exploiting quantum theory in obtaining those results. By analogy, it is instructive to consider whether a complete replacement of a central processing unit (CPU) with a graphics processing unit (GPU) would be sensible. The clear answer is no, and by the same reasoning, one should not expect to simply replace a CPU with a quantum processing unit (QPU). Instead, as was the case during the emergence of modern GPUs, standard information processing tasks must be re-envisioned to take advantage of the new computational resource.

Various strategies have been proposed for applying quantum computers to existing research problems, yet none have demonstrated clear utility in addressing the practical challenges of significant scale or complexity. Many computational tasks remain better suited for CPUs than for GPUs, with the latter excelling in particular domains but not universally. Similarly, while numerous theoretical results predict quantum speedups for specific problems, it is unclear whether such advantages can be realized with near-term hardware. Nevertheless, the prospect of achieving demonstrable quantum advantage continues to motivate substantial research investment, with the potential for transformative impact across multiple domains \cite{xprize2025quantumapps}.


\subsection{Misuse of Foundational Terms}

Existing engineering efforts to exploit quantum technology primarily rely on variational and heuristic strategies, most of which are related to quantum annealing \cite{Dalzell_2025}:

\begin{itemize}
\item \textbf{AQC (Adiabatic Quantum Computing)} denotes the algorithmic model based on the adiabatic theorem, which requires sufficiently slow evolution to preserve the ground state throughout computation.
\item \textbf{QA (Quantum Annealing)} denotes the heuristic approach implemented on physical devices such as D-Wave systems, which may not satisfy strict adiabatic conditions.
\item \textbf{QAC} is an ambiguous acronym denoting "Quantum Adiabatic Computation" \cite{kieu2004hypercomputation, aharonov_adiabatic_2008}, "Quantitative Adiabatic Condition" \cite{li2014quantitative}, "Quantum Annealing Computing" \cite{Nasa2016QAC}, "Quantum-Computing Aided Composition" \cite{costa2022qac}, and potentially other terms.
\end{itemize}

The distinctions between these terms reflect fundamental differences in computational models and expected performance characteristics \cite{quantum_annealing_vs_gate_2024,adiabatic_cyber_2025}. When researchers conflate them, they may inadvertently make theoretical claims about adiabatic guarantees while actually implementing heuristic annealing approaches, leading to incorrect performance expectations and flawed experimental designs.

The historical record further compounds this confusion. Many papers incorrectly describe the relationship between quantum annealing and adiabatic quantum computing, often portraying quantum annealing as a ``specialized variant'' of AQC. This reverses the actual historical and theoretical relationship. Quantum annealing (QA) was introduced in the early 1990s by Finnila et al. \cite{finnila1994quantum} and again by Kadowaki and Nishimori \cite{kadowaki1998quantum} as a quantum analogue of simulated annealing, focusing on optimization through thermal and quantum fluctuations. Preceding each of these is a 1989 paper introducing the idea of QA without coining the term "Quantum Annealing" \cite{apolloni1989quantum}. While it is true that the adiabatic theorem, exploited by AQC was first proposed in 1928 by Max Born and Vladimir Fock \cite{born1928beweis}, this method of quantum computing was not adopted until much later. In 2000, Farhi et al. \cite{farhi2000quantum} proposed Adiabatic Quantum Computing (AQC) as a model for universal quantum computation through adiabatic evolution. Properly understood, AQC represents the ideal, closed-system, adiabatic case of the broader quantum annealing framework. When foundational terms are misused, subsequent research builds on faulty conceptual foundations, potentially wasting years of effort on approaches that cannot deliver their promised advantages.

\subsection{Claims of Computational Advantage}

Beyond terminological precision lies a more critical issue: \emph{the quantum-industrial complex}. We use this term to describe the tendency toward overstated claims about computational advantage, often driven by funding pressures and commercial incentives. A pervasive misstep involves suggesting that quantum approaches offer "advantages" without sufficient empirical evidence \cite{google_dwave_speedup_2016,dwave_quantum_advantage_2025}. This manifests itself through overstated algorithmic speedup claims, selective benchmark reporting, inadequate disclosure of experimental limitations, and commercial promotion of unverified capabilities. All of these undermine research credibility and impede sustained progress.

A fundamental misunderstanding within the quantum-industrial complex concerns the role of theory in algorithm development. Current quantum computers are analogous to early mainframes: just as punch cards required theoretical context to be useful, quantum algorithm development demands foundational, theory-based intuition. For instance, the proven polynomial equivalence between adiabatic quantum computing (AQC) and the circuit model remains underappreciated among engineers. Similarly, quantum state tomography, the classical reconstruction of a full quantum register, scales exponentially with system size. As a consequence, it is computationally prohibitive for large systems. These examples illustrate that theoretical foundations are not academic luxuries, but essential prerequisites for realizing quantum potential.

The quantum annealing community exemplifies these challenges in establishing fair benchmarks \cite{jiang2018quantum, maezawa2019toward, peng2019factoring}. D-Wave systems have demonstrated impressive speedups in controlled contexts, yet these advantages often vanish when compared against optimized classical algorithms rather than naive implementations \cite{google_dwave_speedup_2016}, or when runtime complexity is carefully analyzed \cite{aaronson2013quantum, ronnow2014defining}. Meaningful progress requires comparing quantum solutions against state-of-the-art classical counterparts to advance both fields simultaneously.

While quantum algorithm theory typically provides rigorous asymptotic results, concerns arise more frequently in applied and applications-oriented research. A significant oversight is the tendency to understate or omit practical obstacles that could invalidate reported theoretical advantages. Optimistic extrapolations from small-scale demonstrations to future performance often proceed without acknowledging substantial engineering barriers separating proof-of-concept from practical advantage. This optimism, while motivating, risks misleading funding agencies, collaborators, and the broader scientific community about current capabilities and near-term prospects. Such miscommunication fosters skepticism among stakeholders and jeopardizes support for future research.

\section{How Quantum Missteps Manifest}
The missteps outlined in the previous subsections share a common thread. They arise from researchers' tendencies to approach quantum computing problems by directly translating familiar algorithms and problem formulations without reconsidering the fundamental structure of the computational task. 

\subsection{Agent-Based Modeling and Quadratic Unconstrained Binary Optimization (QUBO)}
This pattern is evident in agent-based modeling, where the allure of quantum speedup has led many researchers to reduce complex social systems to standard optimization frameworks like QUBO without questioning whether such reductions preserve the meaningful structure of the original problem \cite{cmu_qubo_2025,penalty_free_qaoa_2024,qubo_criticism_2024,qubo_constraints_matlab_2024,qubit_efficient_qubo_2023}.

The QUBO reduction represents one of the most pervasive examples of this misguided approach. Many quantum computing practitioners have adopted the belief that any optimization problem must be transformed into a QUBO formulation to be suitable for quantum solving \cite{cmu_qubo_2025,qubit_efficient_qubo_2023,penalty_weights_qubo_2025}. Unfortunately, these systems can only handle unconstrained optimization problems, and there are no real-world problems without constraints. The standard approach of converting constrained problems to QUBO using penalty methods creates terrible landscapes where, for large problems, one may be lucky to find any feasible solution \cite{qubo_criticism_2024}.



\subsection{A Concrete Example of What Not To Do}
Consider the typical approach to quantum agent-based modeling. Researchers begin with a classical ABM, such as Schelling's model, then identify computational bottlenecks (like convergence time or equilibrium detection). Then, they transplant the same structure of the classical implementation into a quantum algorithm, assuming it will solve the underlying problem more quickly. This approach invariably leads to a QUBO reduction which encodes the research question in a traditional agent-based model implementation \cite{cmu_qubo_2025,penalty_free_qaoa_2024,qubo_criticism_2024}. 

We know this because we did it. Several authors of this paper initially tried to design a quantum algorithm to compute the number of agent moves required for global satisfaction in Schelling's model via QUBO reduction. In fact, due to the many links between Schelling's and Ising models in the literature, it is natural to attempt to apply QUBO to the Schelling model \cite{stauffer2008social, laciana2011ising, gauvin2024ising, stauffer2007ising}.

Our workflow began by simulating the traditional implementation of Schelling's segregation model, where agents of two types occupy positions on a grid and their satisfaction depends on neighboring agents. Each agent on the grid was assigned a variable representing its type, red or blue, which was encoded as a binary value: 0 for red and 1 for blue. Every grid position was mapped to a binary variable, so the entire grid was represented as a vector encoding the global state of all agents. The process then constructed a Hamiltonian energy function, assigning higher energy (penalizing) arrangements to agents with neighbors of the opposite type, and lower energy (rewarding) arrangements to agents surrounded by similar types. 

This Hamiltonian fully captured Schelling model's classical social preference rules in mathematical form. Additional terms were incorporated to enforce that only valid agent assignments appeared in final solutions, ensuring the QUBO returned physically meaningful and socially optimal arrangements. This energy function was transformed into a quantum objective, and the Quantum Approximate Optimization Algorithm (QAOA) was applied to find low-energy configurations. QAOA alternated quantum circuit layers representing the problem’s cost structure and mixing, with parameters optimized by a classical optimizer, leading to a quantum state where measurement yields grid arrangements that best reflect the Schelling model’s social preferences.  Our source code of this original attempt at a quantum approach is available here \cite{gore2025schellingqubo}.

This approach is fraught with issues. It destroys the structure that makes the Schelling problem interesting. It cannot be executed on any resources we have access to for more than 20 agents. Finally, as shown in Table \ref{tab:quantum_slowdown}, it does not provide speedup. Ignoring the time required to execute our Qiskit simulation, the problem encoding itself requires more time than executing the classical algorithm. This is an example of exactly what not to do when developing a quantum algorithm for an agent-based model. We discuss each of these issues in further depth in Section \ref{subsec:count-first}.






\begin{table}[htbp]
\centering
\caption{Performance Comparison: Classical Schelling vs. Quantum QUBO Implementation Problem Encoding}
\label{tab:quantum_slowdown}
\begin{tabular}{crrrrrr}
\toprule
Grid & Agents & Qubits & Classical & Quantum Encoding & Classical Speedup\\
Size &  & Required & Time (ms)  & Overhead (ms) & Over Problem Encoding  \\
\midrule
3$\times$3 & 7 & 18 & 1.41 & 64.15 & 45.82x \\
4$\times$4 & 12 & 32 & 2.46 & 1,631.81 & 663.34x\\
\bottomrule
\end{tabular}
\end{table}

\subsection{Understanding Quantum Advantage Through the Welded Tree Problem}

Our failure with using QUBO as a means to create a traditional ABM implementation reflects a fundamental misalignment between classical ABM approaches and quantum computational paradigms. In particular, most classical ABM approaches implement Markov processes. A time-homogeneous Markov chain on a finite space $\mathcal{S}$ is a family of random variables $X_0,X_1,X_2,\dots$ such that for all $t \in \mathbb{N}$ and $u \in \mathcal{S}$
\[
    \mathrm{Pr}\left(X_{t+1} = u \; \vert \; X_t, X_{t-1},\dots, X_0\right) = \mathrm{Pr}\left(X_{t}=u \;\vert\; X_t \right)
\]
for all $u \in \mathcal{S}$ and $t \in \mathbb{N}$. 

In other words, the state $X_{t}$ ``screens off'' information held in all states $X_{t' <t}$. While we can always refer to a long history, only the most recent state helps us understand the actual behavior of the system. In a Markov process, all behavior is determined by (1) the current state of the system and (2) the transition rules of the system. Formally, we can define a \textit{transition operator} $U:\mathcal{S} \rightarrow \mathcal{S}$ such that if $\pi_{t}(x)$ represents $\Pr{X_{t+1} = x}$, then $\pi_{t+1} = U \pi_t$. As we normally deal with finite systems, we can formalize this as a matrix. However, it need not be implemented as such for the mathematical machinery to still apply. This point is crucial to understanding how a quantum computer achieves computational advantage.

To clarify these concepts, we examine the welded tree problem \cite{childs2003exponential, childs_quantum_forgetting_2023}, a Markov process that can, in principle, arise within an agent-based modeling context. This problem is especially significant because it demonstrates a well-known exponential separation between classical and quantum computational performance. Notably, the welded tree is a specific instance of a more general class of problems we consider in the context of the Schelling model; thus, the existence of exponential quantum speedup is guaranteed for certain agent-based models.

The welded tree problem illustrates both the mechanisms underlying genuine quantum speedup and the reasons why standard agent-based modeling algorithms are typically unable to capitalize on these speedups. While exponential improvements do exist, as the welded tree shows, realizing them in ABMs requires an input-output perspective that we describe in later sections.

\subsubsection{The Welded Tree Problem Query Model}

In the welded tree problem, two complete binary trees of height $n$ are connected at their leaves through a random bijection, creating a single connected graph with two special vertices: \textproc{Entrance} (the root of the first tree) and \textproc{Exit} (the root of the second tree). The random connections at the leaves form a tangled middle region with potentially extensive cycles. Figure~\ref{fig:welded-tree-example} illustrates this structure.

\begin{figure}[!ht]
\centering
\includegraphics[width=0.8\columnwidth]{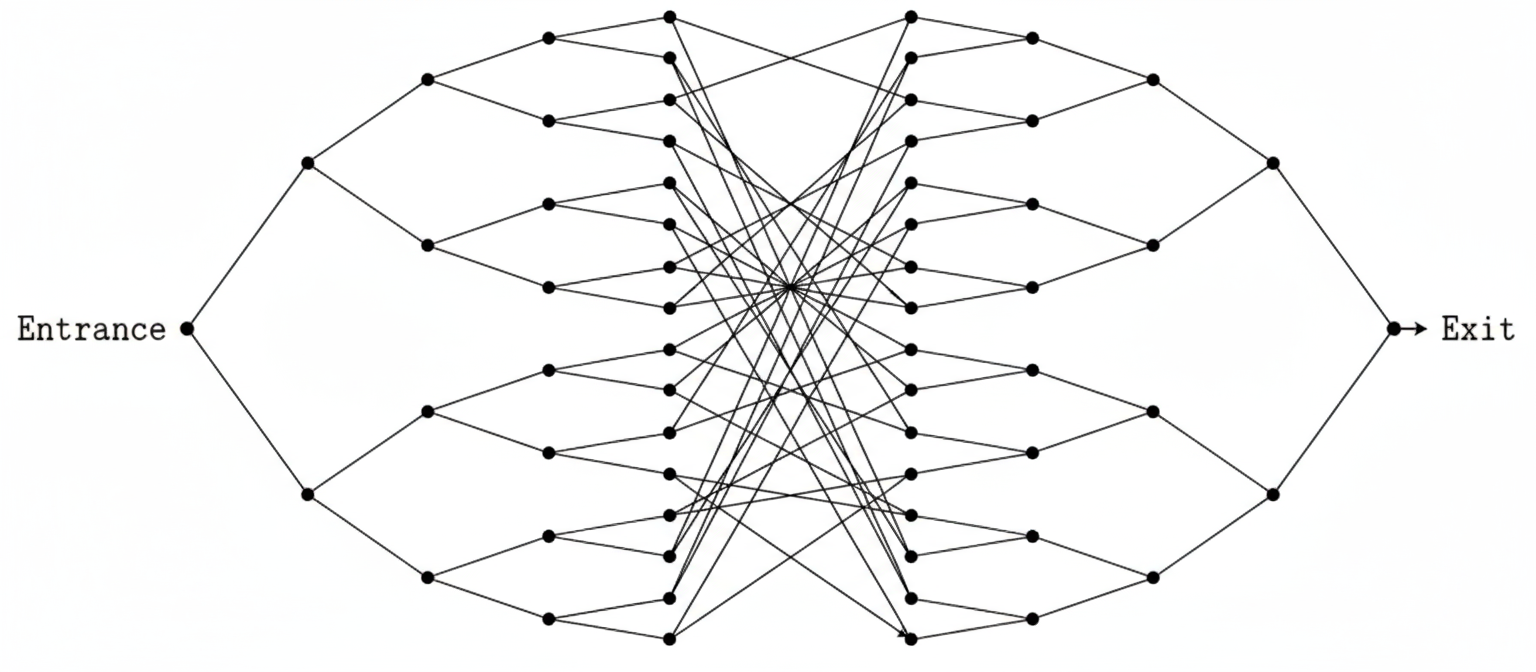}
\caption{An example of the welded tree problem. Two binary trees are welded together at their leaves through random connections, creating a complex middle region. Adapted from \cite{childs_quantum_forgetting_2023}.}
\label{fig:welded-tree-example}
\end{figure}

The computational challenge is specified by its query model( Algorithm \ref{problem:welded_trees}).

\begin{algorithm}[!ht]
\caption{Find the \textproc{Exit}}
\label{problem:welded_trees}
\begin{algorithmic}[1]
\State \textbf{Input:} \begin{enumerate}
    \item An oracle $O$ which, on input a vertex label $v \in V$, returns the list of adjacent vertices in adjacency list form:
\[
O(v) = \{ u \in V \mid (v, u) \in E \}.
\]
    \item An oracle $M$ such that $M(v \equiv \textproc{Exit}) = 1$ and $M(v \not\equiv \textproc{Exit}) = 0$.
\textproc{Exit}.
\end{enumerate}
\State \textbf{Output:} The label $v$ of the target vertex $\textproc{Exit}$. 
\State \textbf{Goal:} Given oracle access to $O$ and the label of \textproc{Entrance}, return the label of \textproc{Exit}.
\end{algorithmic}
\end{algorithm}

Crucially, the problem is \emph{not} simply to find \textproc{Exit}, but to find \textproc{Exit} \emph{within this specific query model}. This distinction is important and can be understood with a simple analogy. Imagine searching through a complex, 3-dimensional maze where you are only allowed to look around one room at a time\footnote{Any finite graph can be embedded in a 3-dimensional space, so this reflects a completely general finite structure.}, with no map or knowledge of the overall layout \textit{a priori}. You can ask for the exits from your current location (which other rooms you can directly reach), you can check whether a particular room is the final goal, and you can bring some stickers with you to mark rooms and doors as you see fit. However, you are not given any clues about the “right” direction to take, and rooms can rotate so that you are left with no sense of global direction.

This restriction is what is described in the query model. It is what makes the problem difficult for classical algorithms. Because binary trees branch towards their leaves, most randomly chosen moves bring you closer to the weld. Even knowing a weld exists doesn't help. This is because the randomness of the connection means that you cannot distinguish advantageous from disadvantageous moves and the weld admits long cycles, so marking your history just prevents you from backtracking along a long path. Effectively, each new room could be just another branch or part of a long, tangled loop. As a result, no classical algorithm can solve this problem efficiently. Instead, they fall back on a running time closer to that of exhaustive search and leading to an exponential number of queries with respect to the height of the trees.

\subsection{Markov Diagrams}
There are many algorithms one could use to solve this problem, some Markovian, some not. If you rely heavily on your stickers (in the analogy) then you might not implement a Markov process. Nonetheless, one might naturally get frustrated and think, ``well, if nothing really matters, let's just try a random walk.'' This would be the most basic Markov process and, in this case, the welded-tree itself would represent its Markov diagram. In particular, we can define the transition operator $U:\mathbb{R}^{V} \to \mathbb{R}^{V}$ such that
\[
\begin{array}{c@{\qquad}c@{\qquad}c}
Uv = \dfrac{1}{\deg v}\displaystyle\sum_{u\sim v} u
&
\text{or equivalently}
&
\Pr(X_{i+1}=u \mid X_i=v)=
\begin{cases}
\dfrac{1}{\deg v}, & u\sim v,\\
0, & \text{otherwise.}
\end{cases}
\end{array}
\]

If we were to directly implement this simple random walk as a Markov process, then the state space would be the space of vertices in Figure \ref{fig:welded-tree-example}. The non-zero ordered pairs $(u,v)$ of the above equation correspond with the edges of Figure \ref{fig:welded-tree-example}. Thus, up to the fact that the state \textproc{Exit} is absorbing and, hence, has uni-directional edges, Figure \ref{fig:welded-tree-example} is also the Markov diagram of the corresponding process. 

Although the Markov diagram and the original graph correspond, this need not always be the case. It is only necessarily the case when we are implementing a simple random walk on the graph as above. For ABMs, we can think at the level of the abstract Markov process (or at the level of the allowed state transitions). This allows complicated scenarios like welded trees to arise in less concrete settings. 

Consider the following example: a collection of $N = 2^n$ agents each maintain a binary state of their location. We represent the collective system state as an $n$-bit string $s \in \{0,1\}^n$, where each bit corresponds to one agent's state. The system begins at state $s = 0^n$ (the \textproc{Entrance}).

At each time step, the system evolves by selecting one agent uniformly at random and flipping their bit. This yields a simple random walk on the $n$-dimensional hypercube graph, where vertices are binary strings and edges connect states differing in exactly one bit. The process terminates when the system reaches state $s = 11\dots 1$ (the \textproc{Exit}). Figure \ref{fig:hypercube} elucidates the system as a hypercube. Importantly, within the hypercube, the hitting time for this walk is $\Omega(2^n)$.

\begin{figure}[!ht]
\centering
\includegraphics[width=0.5\columnwidth]{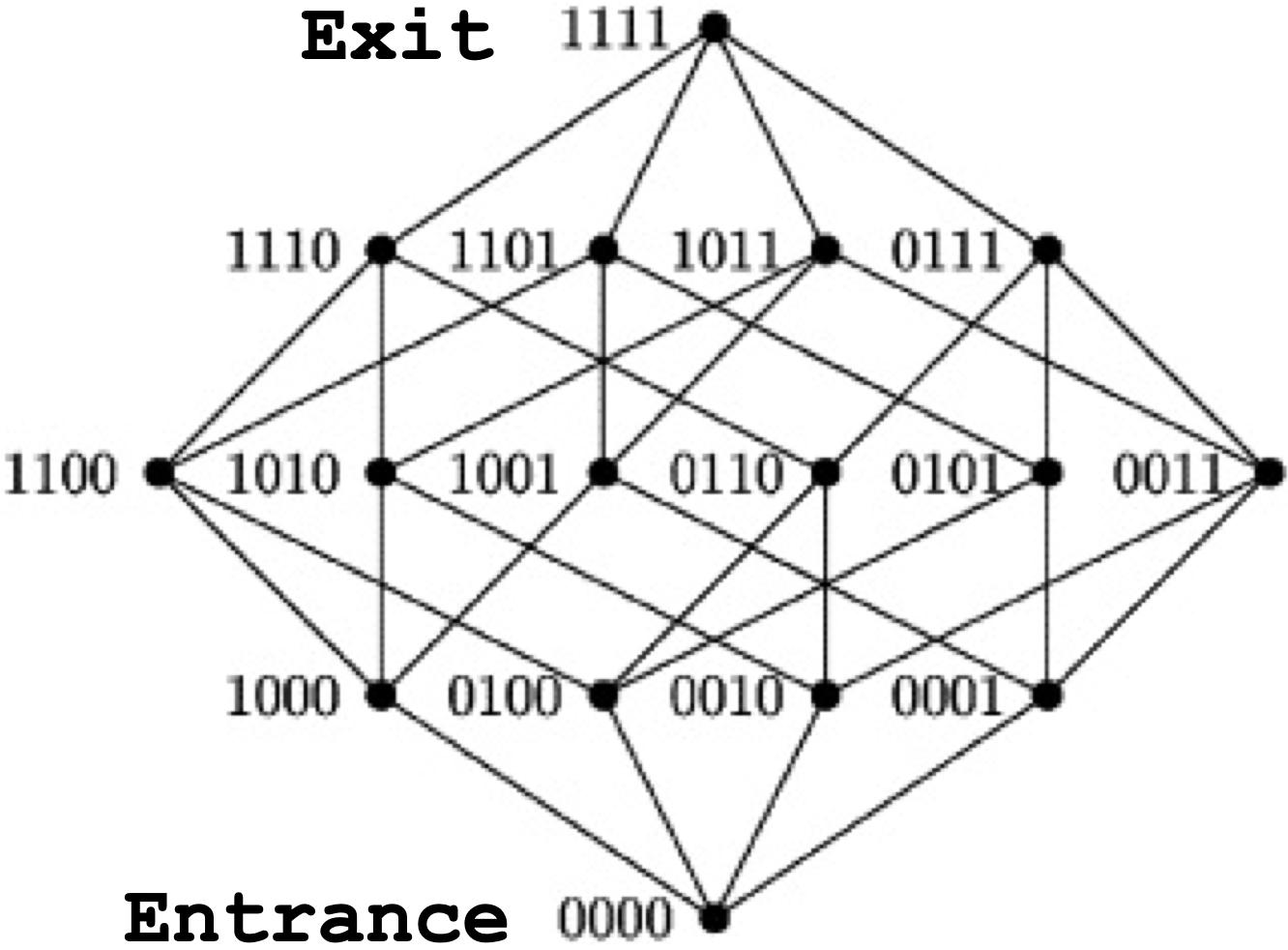}
\caption{$n$-dimensional hypercube graph, where vertices are binary strings and edges connect states differing in exactly one bit.}\label{fig:hypercube}
\end{figure}

As a concrete instantiation, imagine agents arranged in a physical space divided into two sides of a room like in the game dodgeball. Each agent independently and randomly decides whether to switch sides at each time step. Although the agents move in physical space and make individual decisions, the \textit{Markov process} describing the aggregate system state lives on the hypercube graph structure.

The \textit{Markov process} has states corresponding to the $n$-bit strings, with transition probabilities:
\[
\Pr(s_{t+1} = s' \mid s_t = s) = 
\begin{cases}
\frac{1}{n}, & \text{if } s' \text{ differs from } s \text{ in exactly one bit,}\\
0, & \text{otherwise.}
\end{cases}
\]

This defines a Markov chain on the hypercube with $2^n$ states. The Markov diagram is precisely the $n$-dimensional hypercube graph shown in Figure \ref{fig:hypercube}. Essentially, the Markov and agent-based processes occupy different, abstract spaces. Although agent behavior generally induces dependencies between these spaces, they are fundamentally different and related only by the underlying mathematics.


Classical and quantum algorithms can exploit the structure of the abstract Markov state space. The concern is \textit{what computational operations are permitted} and \textit{what information must be observed at each step}. This is especially important when thinking about traditional ABM implementations.

Consider the following two approaches to analyzing the dodgeball system to determine the number of steps until the system reaches the exit state.

\textbf{Approach 1: Traditional Agent-Based Model.} The standard agent-based modeling workflow maintains a complete assignment of agents to states. That is, for a set of agents $\mathcal{A}$ and a set of states $\mathcal{S}$ those agents can occupy, we keep a full specification of a function $f: \mathcal{A} \to \mathcal{S}$. A complete truth-table matrix for this function always requires at least as many bits as $\lvert{\mathcal{A}}\rvert\cdot\textproc{size\_of}(\mathcal{S})$. Until termination, we loop over the following steps:
\begin{enumerate}
    \item The complete system state $f[\mathcal{A}]$ is observed.
    \item An agent is selected and its state is updated $\mathcal{A}\rightarrow \mathcal{A'}$.
    \item The new complete system state $f[\mathcal{A'}]$ is observed.
\end{enumerate}

This approach requires $O(T)$ observations of the complete system state, where $T$ is the number of time steps until the exit state is reached. This workflow requires tracking \textit{which specific agent} has \textit{which specific bit value}. While the Markov state is simply the $n$-bit string (e.g., ``0101''), the traditional ABM must additionally maintain the agent-to-bit mapping: ``Agent 1 has bit 0, Agent 2 has bit 1, Agent 3 has bit 0, Agent 4 has bit 1,'' This information is necessary because the implementation selects a specific agent to update at each step and observe the result.


\textbf{Approach 2: Direct Computation Without Iterative Observation.} An alternative approach asks a different question: ``What is the expected number of steps $T$ until the system reaches the exit state?'' This question does not require simulating the system step-by-step with observation at each iteration. It admits solutions that compute the answer directly from structural properties of the hypercube and the transition probabilities.

A classical analyst might recognize that the structure of the problem is a random walk on a hypercube. They would apply known results from Markov chain theory to compute $T$ in closed form. A quantum algorithm might use a quantum walk to compute spectral properties of the transition operator without ever observing intermediate states. Both approaches work at the level of the abstract Markov state space, but neither requires the iterative observation that characterizes traditional ABM implementation. The difference between these two approaches leads to important consequences in applying quantum algorithms to ABMs.

A quantum algorithm based on a quantum walk can solve the welded tree problem  in polynomial time \cite{childs2003exponential}. The quantum walk explores the graph in superposition, effectively traversing all paths from \textproc{Entrance} to \textproc{Exit} simultaneously. Because both endpoints possess unique symmetry properties within the graph structure, quantum interference constructively reinforces their amplitudes, while destructively canceling paths that lead elsewhere. The walk effectively collapses the exponential structure into a simple path, as illustrated in Figure~\ref{fig:welded-tree-reduced}.

\begin{figure}[!ht]
\centering
\includegraphics[width=0.9\columnwidth]{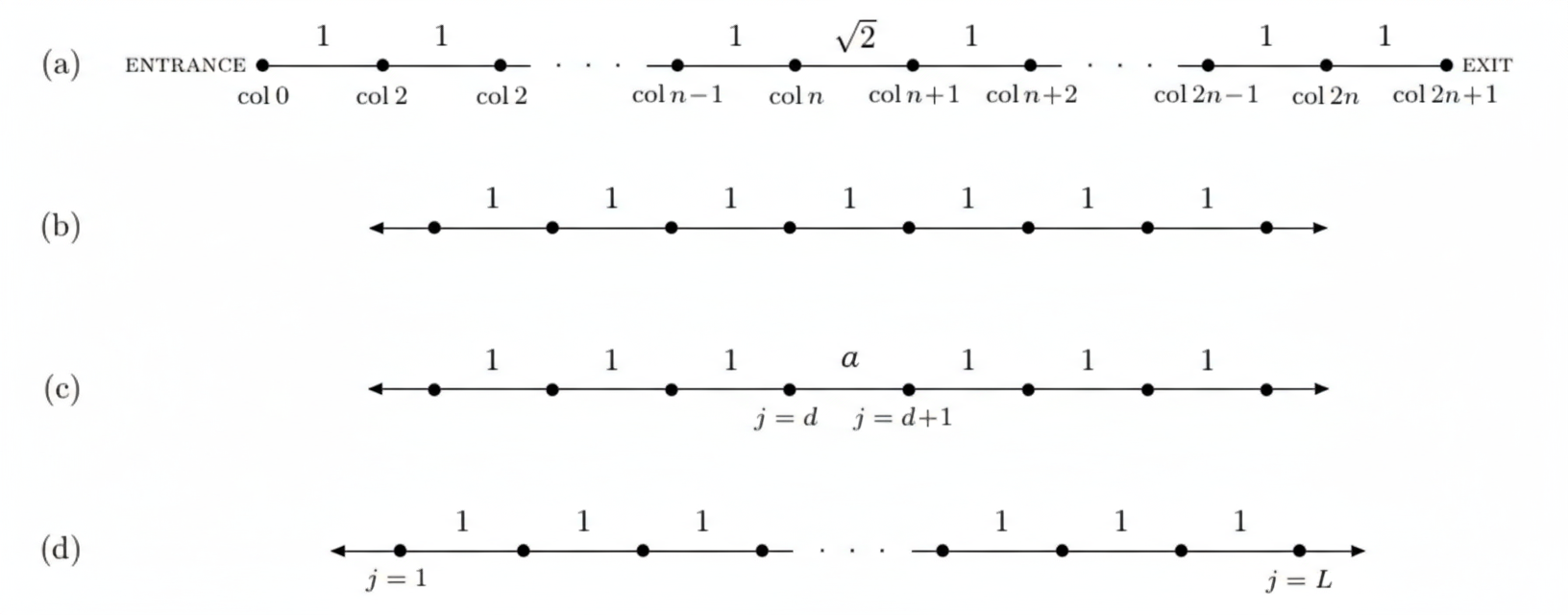}
\caption{Quantum walks effectively reduce the complex to a simple path with $2n+2$ vertices—one per level. The quantum superposition exploits symmetry to propagate efficiently without explicitly tracking individual paths. Adapted from \cite{childs2003exponential}.}
\label{fig:welded-tree-reduced}
\end{figure}

What makes the quantum algorithm powerful is not that it retains more information than traditional ABM implementation. Instead, a quantum walk succeeds by not retaining this information. If one measures the quantum state mid-computation to observe it, the superposition collapses and the quantum advantage vanishes. The quantum algorithm must traverse the entire graph \emph{without observing intermediate states} to preserve the interference patterns that enable polynomial-time solutions. This is critically important. The algorithm discards path information that a traditional ABM implementation would retain.

Knowing the structure of \Cref{fig:welded-tree-reduced} in advance offers a significant opportunity for algorithmic efficiency, as later sections will demonstrate. Unused structural information in a simulation is a missed chance to develop more efficient algorithms. Quantum algorithms are inherently sensitive to such structure and, as the welded tree example shows, often exploit it automatically. While this is a tremendous asset when the structure is unknown, failing to leverage known structure is a misstep—on both classical and quantum fronts.

\subsection{Implications for Agent-Based Modeling}

The previous examples reveal a fundamental incompatibility between quantum advantage and traditional ABM implementations. The power of quantum algorithms emerges from maintaining superposition across computational steps. This is precisely what ABM implementations systematically destroy through their iterative observation of intermediate states along the way to the solution.

To make this contrast explicit, consider the traditional ABM implementation in \Cref{alg:traditional_ABM}.
\begin{algorithm}[!ht]
\caption{Traditional Simulate-Agent-Based Model Algorithm \label{alg:traditional_ABM}}
\begin{algorithmic}
\Function{Simulate}{$s_i,\tau$}
\While{$\tau(s_i) = \textproc{False}$}
\State Do some user I/O with $s_i$
\State $s_{i+1} = Us_i$
\State $i\gets i+1$
\EndWhile
\State\Return
\EndFunction
\end{algorithmic}
\end{algorithm}

The user interaction at line 2 introduces intermediary state observation, collapsing the quantum superposition, and precluding advantage. To correct for this and align with the desired I/O abstraction, we adopt a more precise formulation based on a Markovian model of state propagation. In this framework \textit{getting the next state} is formalized as \Cref{alg:markov_ABM}.
\begin{algorithm}[!ht]
\caption{Get Next State \label{alg:markov_ABM}}
\label{prob:ABM}
\begin{algorithmic}
\Function{Get\_Next\_State}{$s_i \in \mathcal{S}, U$}
\State Take as input An initial state $s_i \in \mathcal{S}$ and an oracle $U : \mathcal{S} \to \mathcal{S}$ specifying state updates.
\State \Return The state $U s_i = s_{i+1}$.
\EndFunction
\end{algorithmic}
\end{algorithm}

Unless a single application of \(U\) is computationally costly, the runtime of \Cref{alg:traditional_ABM} is dominated by how many times \Cref{alg:markov_ABM} is invoked, not by the internal complexity of applying \(U\) itself. In effect, quantum behavior can only occur within \Cref{alg:traditional_ABM}. Therefore, the algorithm achieves no intrinsic speedup. It merely replaces each classical state update with an equivalent quantum operation, yielding the same number of calls. As a result, if the time-per-call is not the dominant behavior and the number of calls is, then any advantage one buys from a quantum computer is extremely limited. This is formalized by the following well-known no-go theorems:
\begin{theorem}[Output-size lower bound \cite{sipser1996introduction}]
\label{thm:output-size}
Any algorithm that must reveal (print, transmit, or otherwise measure) $T$ intermediate
states incurs $\Omega(T)$ time just to produce that output, independent of the internal
computational model.
\end{theorem}
This yields the following statement precisely about quantum algorithms.
\begin{theorem}[Measurement budget and coherence \cite{klm}]
\label{thm:no-go}
If a workflow requires a measurement of the full system state after each of $T$ updates,
a quantum implementation cannot preserve coherence across those $T$ updates. Any
possible quantum speedup must therefore come from reducing the \emph{number of required
observations} or from reformulating the question to avoid them.
\end{theorem}

Theorem \ref{thm:no-go} establishes that if updating the system state of an ABM is computationally inexpensive, taking constant time per step ($O(1)$), then no quantum approach can provide any relevant speedup. This is because the repeated measurements needed to track the system’s state at each step collapse the quantum superposition, destroying the conditions required for quantum advantage. Even in cases where each state update is more computationally costly, taking time proportional to the full simulation length ($O(T)$), the maximum possible quantum speedup is limited by the simulation length. 

The root cause is architectural. ABMs are traditionally implemented to observe the full system state at every time step. This enables bottom-up rules to be applied to the agents based on their specific state at a given time step. This design principle is at the core of canonical agent-based modeling \cite{wilensky2015introduction, epstein1996growing, macal2010tutorial}.  However, it is precisely this repeated observation that precludes quantum advantage. The welded tree and dodgeball problems succeed because they ask a question that requires only the final answer (\textproc{Exit}) rather than the complete trajectory from (\textproc{Entrance}) to (\textproc{Exit}) \cite{childs_quantum_forgetting_2023}. In contrast, the majority of ABMs are built around trajectory observation as the means to calculate the output. As we articulate below, many questions do not require a complete trajectory to answer.

This does not imply that quantum approaches to agent-based modeling are impossible. Rather, it demands that researchers fundamentally re-conceptualize what questions to ask about agent systems and what forms of answers quantum (and even classical) algorithms can meaningfully provide. Aside from edge cases, quantum advantage will not emerge from simulating classical ABM dynamics faster, but from formulating entirely new questions about agent systems.\footnote{It is possible, for instance, that agents might be required to solve problems before making decisions. If one were to introduce this sub-task, then a quantum computer might be able to help the agent actually decide something previously undecidable and confer an advantage even in the traditional setting. Nonetheless, this feels contrived.} These questions do not require observing complete system states at intermediate time steps. Instead, they require reformulating their research question into a problem whose structure aligns with quantum computational paradigms. As we show in the next section, this reformulation alone often reveals new and better classical approaches that can be near-optimal for specific tasks.

\section{Using A Structural Reconceptualization Approach To Avoid Quantum Missteps}
Through discussions between the quantum- and ABM-trained co-authors, we recognized that our initial QUBO formulation for “running an agent-based Schelling's model using a quantum reduction to determine the number of moves required to achieve global satisfaction” introduced the issues we have previously described.

Fully appreciating the issues with our approach, our quantum-trained co-authors helped us rethink the problem’s structure and formulation. Rather than focusing on recasting the traditional Schelling's model into the QUBO framework, we began asking a more targeted question: “How many agent moves are required for global satisfaction in the Schelling model?” Although this question may seem similar at first glance, it is fundamentally distinct. Crucially, it admits a well-defined input/output formalization that aligns better with both theoretical analysis and quantum implementation

\begin{algorithm}[!ht]
\caption{Number of Agent Moves to Global Satisfaction in Schelling's Model}
\begin{algorithmic}[1]
\State \textbf{Input:}
\begin{itemize}
    \item A lattice or network of $N$ agents, each assigned a binary type $\sigma_i \in \{A, B\}$.
    \item A neighborhood function $\mathcal{N}(i)$ specifying which agents are considered neighbors of agent $i$.
    \item A tolerance parameter $\tau \in [0,1]$ defining local satisfaction:
    \[
      i \text{ is satisfied} \iff \frac{\#\{ j \in \mathcal{N}(i) : \sigma_j = \sigma_i \}}{|\mathcal{N}(i)|} \ge \tau.
    \]
    \item An update rule $U : \mathcal{S} \to \mathcal{S}$ that selects an unsatisfied agent and moves it to a vacant site (or swaps agents) to increase local satisfaction.
    \item An initial configuration $s_0 \in \mathcal{S}$.
\end{itemize}

\State \textbf{Output:}
An approximation of the expected integer $T$ such that, after $T$ updates, all agents are satisfied:
\[
\forall i,\; i \text{ is satisfied in } s_T.
\]
Equivalently, output the total number of agent moves required on average to reach global satisfaction.

\State \textbf{Goal:}
Determine $T$ given the initial configuration $s_0$ and the update rule $U$.
\end{algorithmic}
\end{algorithm}

These questions may seem closely linked. To some they may initially even appear identical. However, they are fundamentally different in their computational structure. The first question (running an iterative ABM simulation) requires observing the system state at every time step, making it subject to the limitations described by Theorem~\ref{thm:no-go}. The second question (computing the number of moves required) asks only for a single numerical output and does not mandate any particular solution method. While one approach is to explicitly simulate the model step-by-step until satisfaction is achieved, many alternatives exist. A mathematician might derive a closed-form expression based on the network topology and initial configuration. A computer scientist might recognize the problem as equivalent to a known graph-theoretic optimization and apply specialized algorithms. A physicist might identify conserved quantities or symmetries that constrain the answer. An empirical researcher might train a machine learning model on thousands of configurations to predict $T$ without simulation. The critical distinction is that this reformulated question does not require intermediate state observations. Only the final answer matters. This structural change opens the possibility of quantum approaches, avoiding the measurement-induced collapse that dooms traditional ABM implementations.

\subsection{Dimensions of Convergence for the Schelling Model}
A large body of research on Schelling's segregation model focuses on the impact of individual satisfaction thresholds, population density, and group proportions \cite{anastasi2025schelling, agarwal2021schelling}. Satisfaction thresholds define when agents are satisfied in their local environment and have received significant analytical and experimental study. For example, studies have systematically quantified how varying tolerance levels shape patterns of aggregation, with even minor shifts in the threshold leading to dramatic global transitions \cite{singh2009schelling}.

The influence of population density has also been rigorously analyzed, with research revealing how occupancy ratios affect local clustering and global segregation outcomes where agents are satisfied. Researchers have used extensive large-scale simulations to uncover scaling laws relating density to measures of aggregation and demonstrated that some aggregation effects observed in small systems do not generalize to larger, denser settings \cite{singh2009schelling,tsiatas2009population}. Group proportions, the relative sizes of subpopulations in the model, are another well-explored axis. Analytical and simulation results show that varying these proportions shifts both the stability and nature of segregated configurations. Several studies review how asymmetric group sizes and neighborhood composition affect final states and dynamics of segregation where agents are satisfied \cite{anastasi2025schelling}.

Despite this focus, the role of the underlying network structure, beyond the standard regular lattice or fully connected paradigms, remains less examined, particularly regarding its computational complexity and impact on model evolution. While some recent works have touched on network topology effects, including experiments with complex and dense networks, these typically treat structure as a background rather than a fundamental driver of computational complexity \cite{domic2011dynamics}. Some researchers have compared various network structures and noted that topologies, like scale-free networks, can alter convergence and equilibrium, but comprehensive analytical work is lacking \cite{banos2012network}. Recent research is beginning to interrogate the relationship between network topology and the computational demands on phase spaces navigated by the evolving model, pointing to an open need for more systematic exploration and theory-building in this direction \cite{anastasi2025schelling}.

\subsection{Lollipop Networks}

A lollipop network $L_n^m$ consists of two distinct components: a complete graph (clique) of $m$ vertices connected to a path of $n-m$ vertices \cite{lollipop_graph_wolfram,lollipop_graph_wiki}. This seemingly simple structure creates a topology with extreme computational properties that make it an ideal test case for understanding the limits of both classical and quantum approaches to agent-based modeling. Its properties also make it a good case study for other similar, but more realistic social networks, such as the caveman graph~\cite{watts1999networks}. What makes it so compelling is that it exhibits both properties of \textit{density} and \textit{sparsity} in the same graph.

In the context of the Schelling's segregation model, lollipop networks create a particularly challenging computational landscape for determining agent satisfaction dynamics. An example of satisfied agents on a $L_3^6$ lollipop network is shown in Figure \ref{fig:schelling_lollipop_6_3}. The clique portion of the network creates a region where agents have many potential neighbors and can quickly find satisfactory local configurations. However, the path portion severely constrains agent movement options. It creates a bottleneck that can dramatically slow convergence to global satisfaction. Agents positioned on the path have limited neighborhood options and may need to test every location in the linear structure to find an alternative satisfactory positions when they are dissatisfied.  This creates a scenario where determining the number of moves required for global satisfaction becomes computationally intensive via a traditional ABM implementation \cite{worst_case_clique_analysis_2021}.

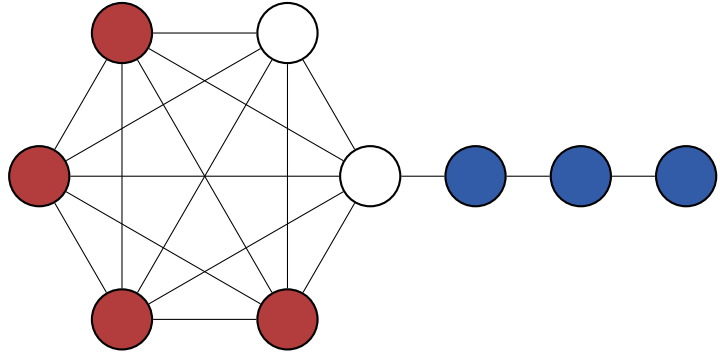
\begin{figure}[!ht]
\centering
\begin{tikzpicture}

\node[circle, draw=black, thick, minimum size=8mm, inner sep=0, fill={rgb,1:red,1;green,1;blue,1}] (c1) at (2.2, 0) {};
\node[circle, draw=black, thick, minimum size=8mm, inner sep=0, fill={rgb,1:red,1;green,1;blue,1}] (c2) at (1.1, 1.905) {};
\node[circle, draw=black, thick, minimum size=8mm, inner sep=0, fill={rgb,1:red,0.70;green,0.24;blue,0.24}] (c3) at (-1.1, 1.905) {};
\node[circle, draw=black, thick, minimum size=8mm, inner sep=0, fill={rgb,1:red,0.70;green,0.24;blue,0.24}] (c4) at (-2.2, 0) {};
\node[circle, draw=black, thick, minimum size=8mm, inner sep=0, fill={rgb,1:red,0.70;green,0.24;blue,0.24}] (c5) at (-1.1, -1.905) {};
\node[circle, draw=black, thick, minimum size=8mm, inner sep=0, fill={rgb,1:red,0.70;green,0.24;blue,0.24}] (c6) at (1.1, -1.905) {};

\draw (c1) -- (c2); \draw (c1) -- (c3); \draw (c1) -- (c4); \draw (c1) -- (c5); \draw (c1) -- (c6);
\draw (c2) -- (c3); \draw (c2) -- (c4); \draw (c2) -- (c5); \draw (c2) -- (c6);
\draw (c3) -- (c4); \draw (c3) -- (c5); \draw (c3) -- (c6);
\draw (c4) -- (c5); \draw (c4) -- (c6);
\draw (c5) -- (c6);

\node[circle, draw=black, thick, minimum size=8mm, inner sep=0, fill={rgb,1:red,0.2;green,0.36;blue,0.66}] (p1) at (3.6, 0) {};
\node[circle, draw=black, thick, minimum size=8mm, inner sep=0, fill={rgb,1:red,0.2;green,0.36;blue,0.66}] (p2) at (5.0, 0) {};
\node[circle, draw=black, thick, minimum size=8mm, inner sep=0, fill={rgb,1:red,0.2;green,0.36;blue,0.66}] (p3) at (6.4, 0) {};

\draw (c1) -- (p1);
\draw (p1) -- (p2);
\draw (p2) -- (p3);

\end{tikzpicture}
\caption{An example of a  $L_3^6$ lollipop network consisting of 4 red and 3 blue agents. White circles denote empty spaces in the network}\label{fig:schelling_lollipop_6_3}
\end{figure}

The pathological nature of these networks - along with the quantifiable ``how many steps'' question - makes them an ideal test bed for understanding the pitfalls of QUBO-like algorithms and the strengths of alternative classical approaches for computing the number of moves required for agent satisfaction in a segregation model. By focusing on the network topology, we developed deeper insights into the structure of the problem. This enabled us to find a new classical solution that is significantly more efficient than running the Schelling's agent-based model to compute the number of moves required for global satisfaction. 

This parallels (to a much less significant extent) the high-profile debate between IBM and Google over ``quantum supremacy.'' Both companies are currently trying to produce quantum devices capable of achieving quantum advantage, or the ability to outperform any conceivable classical computation for a specific computational task. In 2019, Google produced results that demonstrated that they had indeed achieved quantum advantage, albeit on a somewhat contrived and impractical task \cite{arute2019quantum}. In response, IBM showed that careful re-examination and novel classical algorithmic techniques could push back the apparent quantum advantage to a later date \cite{pednault2019leveraging}. Similarly, our work demonstrates that rethinking the formulation of the problem, rather than doggedly pursuing quantum acceleration via reducing a traditional ABM implementation, can yield efficient classical solutions that challenge preconceptions about where quantum speedup is genuinely attainable. Additionally, understanding of these nuances is necessary for designing efficient quantum algorithms for ABM in the future.

Our algorithm provides a concrete performance benchmark that any future quantum approach must surpass on this task. In the next subsection, we describe our new classical approach and compare it to using the traditional Schelling's ABM implementation to compute the number of moves needed for agent satisfaction in a lollipop network.

\subsection{A Count-First Algorithm for Determining Moves to Global Satisfaction in the Schelling Model on Lollipop Networks}
\label{subsec:count-first}

The count-first algorithm exploits structural symmetries in lollipop networks. It divides the network into its three components: (1) the clique, (2) the path, and (3) the bridge (their intersection). One should note that the theorems and proofs contained in the following sections are not up to the typical standards of mathematical rigor, but we have favored accessibility in exchange for rigor and assume a mathematically-minded reader would have no problem adapting them.

\subsubsection{Clique Satisfaction via Global Counts}
\label{subsec:clique}

Before discussing the lollipop itself, it is worth discussing the behavior of a Schelling's model process on the clique, because it reveals the underlying utility of exploiting structure. Assume that the clique is initialized with $\abs{V}$ vertices, $a$ agents of type $A$ and $b$ agents of type $B$. Note that for any agent of type $x \in \{a,b\}$, we have that its local satisfaction is given by
\(
    \textproc{SQ}(x) = \frac{x}{b+a}.
\)
Crucially, this quotient is \textit{entirely} independent of $\abs{V}$. Additionally, we also know (based on this symmetry) that the expected number of steps one can take is precisely $0$ or $\infty$. Thus, in this case, we are only required to decide whether the Schelling process is initially satisfied or never satisfied, which is a traditional decision problem.

We make the following obvious claim:
\begin{theorem}\label{thm:trivial-clique}
    In a clique with $a$ agents of type $A$ and $b$ agents of type $B$, every $A$-agent's other-type fraction is $b/(a+b-1)$ (and vice versa for $B$). Permutations within the clique cannot change this ratio.
\end{theorem}
\begin{proof}
    Although we claim this is obvious, one need only note that, in a clique, all agents are adjacent and repositioning an agent does not change this. Any function that is concerned only with adjacency is therefore permutation invariant.
\end{proof}

As a consequence, the following algorithm can simulate the number of steps to satisfaction on the clique in time $O(1)$. That is, after exploiting symmetry, the count-first approach reduces to \Cref{alg:sim_clique}.
\begin{algorithm}
\caption{\textproc{Simulate\_Clique}$(a,b,v, \tau)$\label{alg:sim_clique}}
    \begin{algorithmic}[1]
        \Function{Simulate\_Clique}{$a,b,V,\tau$}
        \If{$\max\left(a,b\right)\leq \tau \left(a + b -1\right)$} \Return $0$.
        \EndIf
        \State\Return $\infty$. 
        \EndFunction
    \end{algorithmic}
\end{algorithm}

Due to \Cref{thm:trivial-clique}, no computational work needs to be done to answer the question of whether the clique is satisfied. All one needs to do is simply ask whether the system is initially satisfied, which requires $O(1)$ elementary operations. Similarly, a count-first approach on the clique would call $\textproc{Simulate\_Clique}$ and learn immediately that no internal dynamics changing local satisfaction are possible. 

Now, consider the traditional ABM approach. 
\begin{algorithm}[!ht]
\caption{\textproc{Simulate\_Clique\_Traditional}$(a,b,v,\tau,\textproc{runtime})$ \label{alg:sim_clique_trad}}
    \begin{algorithmic}[1]
        \Function{Simulate}{$a,b,v,\tau,\textproc{runtime}$}
        \State Initialize a random array $X \in V^{a+b}$ such that $X_i \neq X_j$ for any pair $(i,j)$.  
        \For{$T\in 0,1,\dots,\textproc{runtime}-1$}            \If{$\textproc{Total\_Unhappy(X)} \texttt{==} 0$} 
                \State\Return $T$
            \EndIf
            \State $i \gets \textproc{Get\_Unhappy\_Index}(X)$
            \State Sample $v \sim \mathrm{Uniform}(V\setminus X)$
            \State $X_i \gets v$
        \EndFor
        \State\Return $\infty$
        \EndFunction
    \State %
        \Function{Is\_Unhappy}{$x \in V$}
        \State $\textproc{Neighbors}\gets 0$
        \For{$y \in N(x)$}                
            \State $\textproc{neighbors}\texttt{++}$
            \State $\textproc{different} \gets 0$
            \If {$\textproc{type}(x) \neq \textproc{type}(y)$}
                \State $\textproc{different}\texttt{++}$
            \EndIf
        \EndFor
        \State \Return $\textproc{different} > \tau \cdot \textproc{neighbors}$
        \EndFunction
    \State %
        \Function{Total\_Unhappy}{$X \in V^{a+b}$}
        \State $\textproc{total} \gets 0$
        \For{$x \in X$} 
            \State $\textproc{total}\texttt{+=} \textproc{Is\_Unhappy}(x)$
        \EndFor
        \State \Return \textproc{total}
        \EndFunction
        \State %
        \Function{Get\_Unhappy\_Index}{$X \in V^{a+b}$}
        \While{\textproc{\textbf{true}}}
            \State Sample $i \sim \mathrm{Uniform}(0,1,\dots,a+b-1)$
            \If{\textproc{Is\_Unhappy}($x_i$)} 
                \State \Return $i$
            \EndIf
        \EndWhile
    \EndFunction
    \end{algorithmic}
\end{algorithm}
Suppose that we label the total running time of $\textproc{Is\_Unhappy}$. Each time we call $\textproc{Is\_Unhappy}$ we make a total of $\abs{N(x)}$ calls to the relevant $\textproc{\textbf{for}}$ loop and, hence, the running time of $\textproc{Is\_Unhappy}$ is $\Theta(\abs{N(x)})$. Now, when computing $\textproc{Total\_Unhappy}$ we require $a+b$ calls to $\textproc{Is\_Unhappy}$, such that its total running time becomes $\Theta\left((a+b)\abs{N(x)}\right)$. Since the $\textbf{\textproc{for}}$ loop in $\textproc{Simulate}$ is called $T+1$ times, we find that the total running time is $\Theta\left((a+b)\abs{N(x)}(T+1)\right)$. 

Now, if we consider explicitly the case of the clique, we have that $N(x) = a+b-1$ always and, if initially the system \textit{is satisfied}, $T = 0$. Hence, the total running time is guaranteed to be $\Theta\left((a+b)^2 \right)$. That is, the traditional method is \textit{polynomially slower} than the count-first method which exploits symmetry. In the case that the system \textit{is not satisfied}, we are forced to execute $\textproc{runtime}$ iterations of the \textbf{\textproc{for}} loop and end up with an algorithm that scales with $\Theta\left( (a+b)^2 \textproc{runtime}\right)$. Note that, depending upon one's choice of $\textproc{runtime}$ this can be substantially worse than $\Theta\left( (a+b)^2\right)$. Additionally, if we do not initially exploit the fact that we only need to decide between $0$ and $\infty$, we are left attempting to empirically approximate the expected running time. A single infinite run produces an unbounded empirical running time, whereas observing a completion time of zero may simply reflect good fortune. Consequently, even in satisfiable cases (ignoring structural considerations), approximating the hitting time requires multiple repetitions.

\paragraph{A QUBO formulation:} Now, we note the impact that this has on QUBO. In particular, QUBO encodes optimization problems and we have, above, been asked to compute a decision problem. The traditional way to rewrite our decision problem as an optimization problem is to ask, for some choice of $T$, "is the expected number of steps prior to satisfaction less than or equal to $T$?" Let us assume that we do exploit the structural consideration that it either is/is not ever satisfied, such that we can choose $T=0$. One can see that an optimum of $T=0$ confirms satisfaction is achieved at that step. However, although we have a success criterion, we lack a direct encoding for a quantum computer. Further reformulation is needed.

The natural approach would be for a QUBO formulation to explicitly calculate how frustrated the existing system is using a Hamiltonian. The Hamiltonian is an operator $H:R\rightarrow R$ where $R$ is some register $R$. For the sake of this manuscript, we allow $R = (0,1,-1)^V$ to be a trinary string representing occupation by agents of various colors. Although we will not focus on the mechanics of how the Hamiltonian implements the cost function in this paper, for a graph $G=(V,E)$ we can define the cost function
\[
    R^\mathtt{T} H_{\textproc{cost}}R = \frac{1}{4}\sum_{v_0 \in V} \sum_{v_1 \in V} \left(R(v_1)^2 R(v_0)^2 - R(v_1)R(v_0)\right) = -\sum_{\{v_0,v_1\} \in E} R(v_0) R(v_1) \left(R(v_0)-R(v_1)\right)^2.
\]
Above, one should note that $H_{\textproc{cost}}$ has the impact of performing the map
\[
    [H_{\textproc{cost}} R](v_i) = \begin{cases}
    1 & \text{if $v_i$ is unhappy} \\
    0 & \text{otherwise.}
    \end{cases}
\]
Above, $H_{\text{cost}}$ is linear. This can be made more obvious by taking the binary mapping $0 \mapsto 00$, $1 \mapsto 01$, and $-1 \mapsto 11$. 

Thus, summing the values in the vector $H_{\textproc{cost}}$ yields the total unhappiness.\footnote{This is also equivalent to the expected value of the occupied subgraph Laplacian under vector $R$, which yields a nicer expression that is more obviously a linear operator. The additional elegance and clarity do not help our exposition here and the current form is more explicit for the uninitiated.} Importantly, implementing this cost function already introduces a pitfall. In particular, \textbf{naively implementing a clique on a register of size $\abs{V}$ requires the ability to simultaneously manage all pair-wise registers}! Although this is possible for very small cliques, as we will see shortly, we do not even have enough interacting qubits yet to create a system that would even come close to being able to implement this cost function at the scale achievable by a desktop computer. 

If this is not bad enough, then we also must consider the Hamiltonian that drives dynamics. This might often be missed, but these dynamics should, indeed, match the dynamics of the corresponding model. For the typical Schelling model, dynamics themselves are given by $H_{\mathrm{driv}}:R\rightarrow R$ and we can assume that, based on our discussion above, we are okay with any \textit{single agent re-assignment}. This is actually non-trivial and extremely important. If we naively allow transitions that are classically forbidden under the corresponding Markov model, we might not simulate the appropriate process. For instance, suppose that agents $A$ are happy no matter what, but agents $B$ need to be adjacent to at least one other $B$ to be satisfied. It is possible to initialize the system, but not the clique, such that the initial configuration determines whether or not satisfaction is even possible. A faithful reproduction of this behavior has to take into account the allowed transitions only.

Now, we know that we $r_{u} = H_{\text{cost}}r$ represents the unhappy agents and we could have similarly defined some $H_{\text{free}}$ such that $r_f = H_{\text{free}}r$ represents the free vertices. Since this computer is quantum and preparing $r_u$ is our only naive means of identifying whether an agent is unhappy, we need to control on $r_u$ (or its complement) before setting whether a vertex is free or unoccupied, another difficult operation with current technology. 

A natural complaint, at this point, would be that the quantum algorithm could also behave in a counts-first sort of way. Although one could, in-principle, encode a less naive version of the algorithm, we would inevitably arrive at the same place. That is, we have always had that
\begin{theorem}
    There does not exist a quantum algorithm capable of deciding whether a Schelling process on a clique will reach satisfaction faster than \cref{alg:sim_clique}.
\end{theorem}
\begin{proof}
    This is just an immediate consequence of \Cref{thm:no-go} and the fact that \Cref{alg:sim_clique} decides this question in $O(1)$ elementary operations.
\end{proof}

\subsubsection{Path Satisfaction via Sequential Scanning}
\label{subsec:path}

The analysis of agents on the lollipop network's path segment also leverages structural properties, but in a distinct manner. Instead of repeatedly querying all agents, we use cached values and local updates to improve efficiency compared to the traditional Schelling approach. Traditionally, \(\textproc{Total\_Unhappy}\) requires querying every agent once per evaluation, which holds true for the count-first method as well. The key difference lies in the frequency of these queries.

For an \(n\)-vertex path encoded as the word \((A,B,0)^V\), a single \(\textproc{Total\_Unhappy}\) query inspects each occupied agent and its neighbors once, incurring \(\Omega(a+b)\) queries. While this cost is similar in both traditional and count-first approaches per query, the count-first algorithm reduces the total number of such queries required by employing \Cref{alg:sim_path}.



\begin{algorithm}[!ht]
\caption{\textproc{Simulate\_Path}$(a,b,v, \tau)$\label{alg:sim_path}}
    \begin{algorithmic}[1]
        \State \textbf{Input:} $a,b,V,\tau$
        \Function{Simulate}{$a,b,V,\tau$}
        \State Initialize a random array $X \in V^{a+b}$ such that $X_i \neq X_j$ for any pair $(i,j)$.
        \State $\textproc{Unhappy\_Count} \gets \textproc{Total\_Unhappy}(V)$
        \State $T \gets 0$
        \While{$\textproc{Unhappy\_Count} > 0$}
            \State $\textproc{Take\_Step}(X,V,\textproc{Unhappy\_Count})$
            \State $T\texttt{++}$
        \EndWhile
        \State \Return $T$.
        \EndFunction
        \State
        \Function{Take\_Step}{$X,V,\textproc{Unhappy\_Count}$}
            \State $i \gets \textproc{Get\_Unhappy\_Index(X)}$
            \State Sample $v \sim \mathrm{Uniform}(V\setminus X)$
            \State Calculate how many agents adjacent to $X_i$ and $v$ are unhappy and let the result be $\textproc{prior}$
            \State $X_i \gets v$
            \State Calculate how many agents adjacent to $X_i$ and $v$ are now unhappy and let the result be $\textproc{after}$
            \State $\textproc{Unhappy\_Count} \gets \textproc{after} - \textproc{prior}$
            \State \Return
        \EndFunction
    \end{algorithmic}
\end{algorithm}
Note that lines 9 and 11 each require precisely $6$ calls to $\textproc{Is\_Unhappy}$ and each call to it is $O(N(v)) = O(1)$. Hence, in this case, $\textproc{Is\_Unhappy}$ is itself an $O(1)$ operation. The benefit is that by only updating the $\textproc{Unhappy\_Count}$, instead of re-calling $\textproc{Total\_Unhappy}$ in line 12, we trade an $\Omega(a+b)$ operation for an $O(1)$ operation. As a consequence, if we assume that we have constant-time access to the adjacency list of each vertex, this algorithm completes in time $O(\abs{V} + T )$ instead of $O(\abs{V}\cdot T)$. That is, simply caching the happiness count trades multiplicative scaling for additive scaling. Whenever the expected value of $T$ scales with $\abs{V}$ ($T \sim \abs{V}^x$), which it will for all but extremely low densities of agents, we have an $O(V^{\max\{1,x\}})$ algorithm instead of an $\Omega(V^{x+1})$ algorithm. That is, the count-first approach always outperforms the standard approach by a factor of $V^{x+1-\max\{1,x\}} = V^{\min\{1,x\}}$.

\subsubsection{Formal Presentation of the Count-First Algorithm}\label{4:algorithm}
The count-first algorithm (\Cref{alg:cf_schel}) exploits structural symmetries in lollipop networks. It divides the network into its three components: (1) the clique, (2) the path, and (3) the bridge. Rather than checking each agent's neighborhood individually (as in traditional ABMs), the algorithm caches values as in \Cref{alg:sim_clique} and \Cref{alg:sim_path}. The algorithm proceeds as follows, where we ignore the nuances of bridge behavior which requires only small updates to bridge behavior when the bridge itself is called. \footnote{The interested reader can find one way to handle these nuances in the specific implementation in \cite{Jarret_Schelling_High-Performance_2025}. We do not anticipate that ignoring this behavior would meaningfully impact the expected value of $T$, due to the all-to-all movement scheme.}

\begin{algorithm}[!ht]
\caption{Count-First Schelling's Model Satisfaction on Lollipop Networks \label{alg:cf_schel}}
\begin{algorithmic}[1]
\State \textbf{Input:} Lollipop graph $G$ with clique size $CS$, path length $PL$
\State \textbf{Input:} Number of agents $n$, satisfaction threshold $\tau = p/q$
\State \textbf{Output:} Number of moves required to reach global satisfaction

\Function{CountFirstSchelling}{$a,b,n,m,\tau$}
    \State Let $V$ be the vertices of a path graph of length $m$ \Comment{$n$ is implicitly the clique size}
    \State $\textproc{Place\_Agents}(a,b,n,V)$
    \State Calculate initial unhappiness of the path and let the result be $\textproc{Path\_Unhappy}$
    \While{$\textproc{Total\_Unhappy} > 0$}
        \State Choose a random number $r \sim \mathrm{Bernoulli}\left(\frac{\textproc{Path\_Unhappy}}{\textproc{Total\_Unhappy}}\right)$.
        \If{$r$}
            \If{$p < n - C_a -C_b$}
                \State \textproc{\textbf{continue}}
            \EndIf
            \State Draw $X\gets A$ with probability $\frac{C_A \cdot Is\_X\_Unhappy(A)}{\textproc{Clique\_Unhappy}(C_A,C_B)}$ and $X\gets B$ otherwise.
            \State $C_X \texttt{--}$
            \State \textproc{Add\_To\_Path}($X$)
        \Else
            \State Draw $X \sim \mathrm{Bernoulli}\left(\frac{n-C_A-C_B}{n+m-a-b}\right)$
            \If{X}
                \State $\textproc{Simulate\_Path}(X,V,\textproc{Path\_Unhappy})$
            \Else
                \State $i \gets \textproc{Get\_Unhappy\_Index}(V)$
                \State $t \gets \textproc{type}(V_i)$
                \State $C_t\texttt{++}$
                \State $V_i \gets 0$
            \EndIf
        \EndIf
        \State $T\texttt{++}$
    \EndWhile
    \State \Return $T$
\EndFunction
\State
\Function{\textproc{Clique\_Unhappy}}{$C_A, C_B$}
    \State \Return $C_A \cdot \textproc{Is\_X\_Unhappy}(A,C_A,C_B) + C_B\cdot \textproc{Is\_X\_Unhappy}(B,C_A,C_B)$
\EndFunction
\State
\Function{\textproc{Is\_X\_Unhappy}}{$X \in \{A,B\},C_A,C_B$}
    \State\Return $C_X < \tau (C_A + C_B)$
\EndFunction
\State
\Function{\textproc{Total\_Unhappy}}{$C_A,C_B,\textproc{Path\_Unhappy}$}
    \State \Return $\textproc{Clique\_Unhappy}(C_A,C_B) + \textproc{Path\_Unhappy}$
\EndFunction
\end{algorithmic}
\end{algorithm}

\subsubsection{Computational Complexity of the Count-First Algorithm}\label{4:complexity}
We analyze the computational complexity of the count-first algorithm by looking at the per-iteration cost of evaluating satisfaction and selecting moves required to reach global satisfaction. 

In a traditional ABM on a lollipop graph, each iteration requires checking the satisfaction of every agent by examining its neighborhood. Note that, unlike the pure path simulation, clique behavior demands that each member of the clique has $k$ occupants. Then, each recalculation of unhappiness requires $k(k-1)$ operations for the clique and an additional $a+b-k$ operations for the path. Thus, the total number of operations for the traditional approach to calculate $\textproc{Total\_Unhappy}$ is $k(k-1)(a+b-k)$. Note that if $(a+b) \geq m + \epsilon \abs{V}$ the pigeonhole principle guarantees that at least $k \geq \epsilon \abs{V}$ agents occupy the clique at any given time. If we fix any $\epsilon > 0$ independent of problem size, we find that combining these expressions,
\begin{align*}
    k(k-1)(a+b-k) &\geq k(k-1)(a+b-m) 
    \\&\geq k(k-1)(\epsilon\abs{V})
    \\&\geq (\epsilon\abs{V})^3-(\epsilon\abs{V})^2
    \\&= \epsilon^3\abs{V}^3\left(1 - \frac{1}{\epsilon\abs{V}} \right).
\end{align*}
For large enough $\abs{V}$, this is clearly $\Omega(\epsilon\abs{V}^3)$.

In contrast, we have already seen that internal clique dynamics are $O(1)$ and, in fact, the count-first algorithm skips them altogether in Line 13. Thus, we only have path-to-path, path-to-clique, and clique-to-path dynamics. We have already seen that path-to-path dynamics and path-to-clique dynamics are at worst $O(\abs{V})$. Thus, we have the following Theorem.

\begin{theorem}\label{thm:count-first-query}
    Count-First Schelling Satisfaction completes with at most $O(T)$ calls to $\textproc{Get\_Unhappy\_Index}$.
\end{theorem}

\paragraph{Final notes on count-first scaling and possible quantum advantages}
The complexity of identifying unhappy agents depends significantly on implementation details. We store unhappy indices as a bitstring, where a vertex is unhappy if and only if its corresponding bit is set to 1. Under the assumption that each unhappy count appears $O(1)$ times during the algorithm's execution, identifying unhappy bits requires $O(|V| \log |V|)$ time per run. Our empirical scaling matches this bound, suggesting that agent identification constitutes the primary computational bottleneck.

This analysis leaves room for potential quantum speedup via quantum search. A quantum algorithm for unsorted search over $m$ unhappy agents requires time at least $\Omega\left(\sqrt{|V|/m}\right)$. However, the overall running time remains $\Omega(|V|)$, yielding a maximum improvement of only $O(\log |V|)$—a polylogarithmic factor typically considered below the threshold of practical interest and often omitted from formal complexity analysis.

\subsection{Experimental Validation of Count-First vs. Traditional ABM Runtime}\label{4:experiments}

To validate our theoretical complexity analysis, we empirically measured the runtime scaling of both the traditional ABM and count-first implementations on lollipop networks of varying sizes.

\subsubsection{Experimental Setup}

We implemented the traditional ABM in the Mesa agent-based modeling framework~\cite{python-mesa-2020}, optimized to support runs of up to 10,000 agents on arbitrary network topologies. The count-first algorithm was implemented in C++ following the pseudocode presented in Section~\ref{4:algorithm}.

The comparison between traditional and count-first implementations used the following parameters: 80\% agent density, 50\% threshold similarity, and a 50\%-50\% split of agent types. The lollipop network was configured with 10\% of nodes in the clique and 90\% in the path. For each network size, we ran 500 trials and averaged the total runtime to reach global satisfaction. An overview of the wall clock time required for each approach is shown in Figure \ref{fig:lpop_time}.

\begin{figure}[!ht]
\centering
\includegraphics[width=0.95\columnwidth]{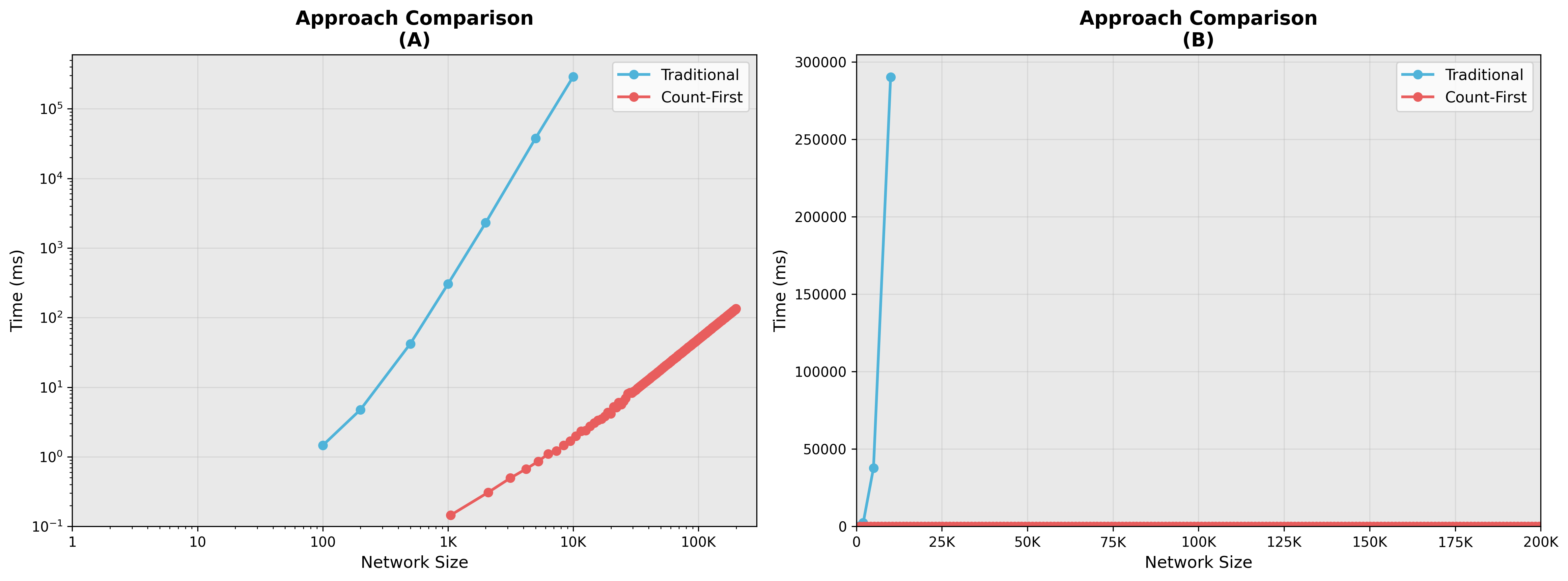}
\caption{(A) Log-log plot of the individual data points of the runtime it takes to compute the number of moves for agents in the Schelling model to reach global satisfaction. (B) The same plot with the x and y axes presented on a linear scale. The approach shown in red is our count-first algorithm. The approach shown in blue is the traditional ABM implementation.}\label{fig:lpop_time}
\end{figure}

\subsubsection{Regression Analysis}

To determine the empirical scaling exponents for each approach, we tested six standard regression models against the experimental data (Tables~\ref{tab:reg_fit} and~\ref{tab:cf_fit}). For the traditional ABM, the power regression model provided the best fit.

\begin{table}[ht]
\centering
\caption{Regression Model Comparisons for Traditional Agent-Based Model}
\label{tab:reg_fit}
\begin{tabular}{lccc}
\toprule
\textbf{Model} & \textbf{Model Equation} & \textbf{RMSE} & \textbf{R\textsuperscript{2}}\\
\midrule
Power & 4.29e-10·n\textsuperscript{3.01}  & 0.037 & 1.000 \\
Quadratic & 4.47e-06·n\textsuperscript{2} &  21.98 & 0.981 \\
Linearithmic & 4.25e-03·n·log(n)  &  59.65 & 0.859 \\
Linear & 0.038·n  &  66.83 & 0.823 \\
Logarithmic & 13.60·log(n) &  147.0 & 0.146 \\
Constant & 74.71 &  159.0 & 0.000 \\
\bottomrule
\end{tabular}
\end{table}

The power regression achieves \(R^2 \approx 1.000\), indicating an excellent fit to $O(n^{3.01})$, which matches the expected scaling described in Section~\ref{4:complexity}. To verify this scaling holds across parameter variations, we tested different agent densities and similarity thresholds (Figure~\ref{fig:power_reg}). The power law exponent remained consistent (ranging from 2.94 to 3.08), confirming that cubic scaling is robust to parameter changes.

\begin{figure}[!ht]
\centering
\includegraphics[width=1.0\columnwidth]{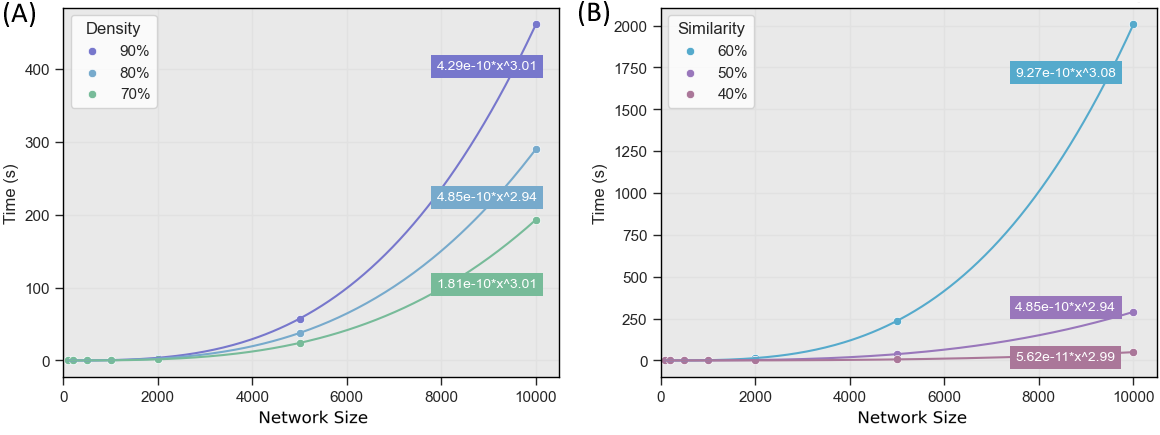}
\caption{Power regression approximations for traditional ABM under parameter variations. (A) Effect of agent density on scaling. (B) Effect of similarity threshold on scaling. The exponent remains consistent across variations, confirming robust cubic scaling behavior.}\label{fig:power_reg}
\end{figure}

For the count-first algorithm, we applied the same fitting procedure (Table~\ref{tab:cf_fit}):

\begin{table}[ht]
\centering
\caption{Regression Model Comparisons for Count-First Algorithm}
\label{tab:cf_fit}
\begin{tabular}{lccc}
\toprule
\textbf{Model} & \textbf{Model Equation} & \textbf{RMSE} & \textbf{R\textsuperscript{2}}\\
\midrule
Power & 5.18e-04·n\textsuperscript{1.05}  & 19.64 & 0.996 \\
Quadratic & 1.25e-09·n\textsuperscript{2} &  139.7 & 0.784 \\
Linearithmic & 7.39e-05·n·log(n)  &  20.06 & 0.996 \\
Linear & 1.03e-03·n  &  21.77 & 0.995 \\
Logarithmic & 39.67·log(n) &  267.9 & 0.207 \\
Constant & 491.28 &  300.8 & 0.000 \\
\bottomrule
\end{tabular}
\end{table}

The count-first algorithm's runtime achieves linear or  near-linear scaling. The RMSE and textbf{R\textsuperscript{2}} show very similar results for power (5.18e-04·n\textsuperscript{1.05}), linearithmic (7.39e-05·n·log(n)), and  linear (1.03e-03·n) scaling. This represents a dramatic reduction compared to the traditional ABM's cubic scaling.

\subsubsection{Scaling Exponent Analysis}
Recall from the previous subsection that both approaches showed the best fit with a power model. Here we quantify the uncertainty in scaling exponents for that model for each approach by three distinct fitting methods to both datasets. This allows us to understand the range of possibilities for the scaling exponent under different assumptions (see Table~\ref{tab:scaling_exponents}).

\begin{table}[ht]
\centering
\caption{Scaling Exponents for Count-First and Traditional ABM Implementations Across Various Fitting Methods}
\label{tab:scaling_exponents}
\begin{tabular}{lcc}
\toprule
\textbf{Fitting Method} & \textbf{Count-First Exponent} & \textbf{Traditional ABM Exponent} \\
\midrule
Polyfit (log-log regression) & 1.347 & 1.914 \\
Nonlinear LS & 1.05 & 2.813 \\
Local Exponent & 1.09 & 2.178 \\
\bottomrule
\end{tabular}
\end{table}

The fitting methods—(1) log-log linear regression, (2) nonlinear least squares, and (3) local scaling exponent via finite differences—provide complementary perspectives on scaling behavior. Log-log regression (Method 1) is robust when the underlying relationship is linear in log-log space. Nonlinear least squares (Method 2) fits in the original scale, allowing larger runtimes to dominate. Local exponent analysis (Method 3) detects point-wise scaling, revealing non-uniformities \cite{pushak2020advanced}.

Across all methods, the count-first algorithm exhibits subquadratic scaling with exponents ranging from 1.05 to 1.347. In contrast, the traditional ABM implementation shows near-cubic scaling with exponents from 1.914 to 2.813. This two-to-three-fold difference in scaling penalties is substantial and compounds with network size.

These empirical results validate the theoretical speedup predicted by our complexity analysis. The traditional ABM could not complete sufficient number of runs beyond 10,000 agents within practical time and memory constraints. Meanwhile, the count-first algorithm continued producing results in .2 seconds (200ms) for 200,000 agents. Even at 1,000,000 agents the algorithm still computes the results in 1 second. This disparity underscores how the count-first algorithm's reduction in complexity scales to problem sizes that have previously been out of scope for the agent-based modeling community.

The counts-first algorithm exhibits slight deviations from perfect linearity, which we attribute to inefficiencies in the data structures used for agent position storage. Standard data structures introduce negligible overhead for moderately sized agent populations; however, at extreme scales, this overhead may become more pronounced, potentially yielding sublinear behavior. Whether such deviations can be mitigated through alternative algorithmic approaches remains an open question that we leave for future work.

\section{Discussion}
Recent efforts to bring quantum computing techniques into agent-based modeling have been intellectually stimulating, but have yet to resolve the core semantic mismatch between classical agent-based models and quantum primitives. QUBO formulations for network optimization, and variational quantum circuits have all attracted attention for their potential to simulate complex social dynamics \cite{qubo_multiagent_2025,quantum_annealing_traffic_2024,quantum_walk_social_2021,quantum_walk_community_2020}. Yet, none of these approaches fully account for the methodological gap: traditional ABM is implemented by iteratively updating agent states based on localized interactions, whereas quantum algorithms depend on coherent, global amplitude evolution, which is fundamentally disrupted by any iterative, stepwise observation~\cite{ozhigov1998quantum,childs_quantum_forgetting_2023}.

At first glance, this seems to argue for abandoning agent-based modeling in favor of more ``quantum-friendly'' approaches. However, this is not the case, and to do would be misguided. Agent-centric models analyze the particularities of individual preferences and local dynamics. Our count-first algorithm is an agent-centric model. It computes, for a given network and agent configuration, whether a globally satisfied solution is reachable, and if so, the precise number of moves required. While iteratively computing the state of a model is not aligned with quantum solutions, agent-centric modeling still can be. 

The fundamental challenge lies in how quantum advantage actually arises. When quantum algorithms do succeed, they succeed because they exploit hidden structure. This structure reflects patterns or symmetries in a problem that remain invisible to classical algorithms. However, once a researcher can identify and describe these hidden structures explicitly, classical algorithms can often exploit it too \cite{PhysRevLett.125.170504, gilyen_10.1145/3406325.3451060}. Consider the lollipop network we study here. Its structure, a dense clique connected to a sparse path, is not hidden. It is precisely defined. Our count-first algorithm exploits this explicit structure to achieve efficiency. This means any quantum approach to the same problem faces the same structural insight, eliminating the very advantage quantum methods might otherwise claim. A key takeaway is that: for a quantum algorithm to achieve advantage here, it must exploit aspects of the problem that are either not classically available or not classically exploitable.

This connects to a deeper architectural point about quantum versus classical information processing. Classical systems, including traditional ABM simulations, track probability. At each step, how likely is the system to occupy each possible state? Quantum systems, by contrast, evolve amplitudes. Amplitudes are mathematical quantities that can interfere constructively and destructively, enabling, for instance, quantum walks to sidestep barriers that trap classical probability-based walks. Yet achieving this requires the quantum system to remain isolated from measurement. The moment a researcher asks "what state is the system in right now?" the quantum amplitudes collapse to classical probabilities and the game ends. Our count-first algorithm sidesteps this problem by abandoning simulating the intermediate states of the system. Instead, it computes directly from structural properties whether satisfaction is achievable, and if so, the exact count of moves required. This is structural analysis. Although a quantum algorithm might inherently exploit this structure without any analysis, to demonstrate the quantum advantage, we must show that a classical algorithm cannot exploit the very same structure under the very same access constraints.

Methodologically, these findings suggest a need for recalibrating how researchers approach agent-based modeling in the era of quantum computation. The process of recalibrating might yield substantially more powerful ABMs even without quantum computers. Attempting to force agent-based modeling problem-solving into quantum frameworks through stepwise state observation is not the answer for increased efficiency. Instead, we need to recognize when classical or quantum algorithms demand a structural rethinking of the research question. Our count-first algorithm is fundamentally bottom-up. However, it achieves computational tractability by leveraging analytically tractable features of the system, moving beyond naive and direct simulation. The value of ABM lies precisely in its granularity and flexibility. The lesson from our research is that efficiency sometimes requires reframing the problem at the structural level. Doing so does not always lose flexibility; it sometimes allows us to answer previously intractable questions.

\section{Conclusion}
This work challenges the prevailing approach to quantum agent-based modeling. Rather than forcing agent-based modeling problems into quantum frameworks, we must ask fundamentally different questions that align with quantum computational capabilities. Our journey illustrates this principle through concrete missteps and successes. Our initial reduction of Schelling's model to QUBO formulation exemplifies the wrong direction. The direct translation obscured what makes agent-based systems interesting while failing to leverage quantum advantage. In contrast, our count-first algorithm demonstrates the right approach. It is a structural reconceptualization that exploits the explicit problem architecture to compute minimum moves to global satisfaction on lollipop networks, establishing a concrete lower bound a quantum algorithm must breach for advantage. Any quantum method must overcome this baseline, plus the substantial overhead of merely running the problem on a quantum device. Given the current and near-term state of quantum hardware, this is unlikely. 

However, this conclusion should not discourage quantum approaches to agent-based systems. Rather, it clarifies what quantum advantage requires. Quantum advantage requires problems whose computational structure aligns naturally with quantum primitives: the ability to exploit interference, superposition, and entanglement. In general, it cannot be achieved by reducing traditional ABM implementations to well-known quantum problems. By embracing structural reconceptualization as a methodology, researchers can identify ABM questions where quantum approaches might genuinely outperform classical solutions. At the same time, researchers adopting this approach will advance classical understanding through structural insights, and even if quantum computers never come to fruition, their efforts will still have advanced science.

\section*{Author Contributions}

The following author contributions are categorized in \Cref{tab:contributions} according to the \textit{CRediT (Contributor Roles Taxonomy)} \cite{credit2022contributor}. The author order in this paper is strictly alphabetical and does not imply relative levels of contribution.
\begin{table*}[!ht]
\centering
\renewcommand{\arraystretch}{1.2}
\resizebox{\textwidth}{!}{%
\begin{tabular}{lccccccccccccc}
\toprule
\textbf{Author} &
\rotatebox[origin=c]{60}{\parbox{3.2cm}{\centering Conceptualization}} &
\rotatebox[origin=c]{60}{\parbox{3.2cm}{\centering Formal Analysis}} &
\rotatebox[origin=c]{60}{\parbox{3.2cm}{\centering Investigation}} &
\rotatebox[origin=c]{60}{\parbox{3.2cm}{\centering Methodology}} &
\rotatebox[origin=c]{60}{\parbox{3.2cm}{\centering Resources}} &
\rotatebox[origin=c]{60}{\parbox{3.2cm}{\centering Software}} &
\rotatebox[origin=c]{60}{\parbox{3.2cm}{\centering Visualization}} &
\rotatebox[origin=c]{60}{\parbox{3.2cm}{\centering Writing – Original Draft}} &
\rotatebox[origin=c]{60}{\parbox{3.2cm}{\centering Writing – Review \& Editing}} &
\rotatebox[origin=c]{60}{\parbox{3.2cm}{\centering Funding Acquisition}} &
\rotatebox[origin=c]{60}{\parbox{3.2cm}{\centering Project Administration}} &
\rotatebox[origin=c]{60}{\parbox{3.2cm}{\centering Supervision}} &
\rotatebox[origin=c]{60}{\parbox{3.2cm}{\centering Validation}} \\
\midrule
C. Nico Barati        & X &   &   &   &   &   &   &   &   &   &   &   &   \\
Arie Croitoru         & X &   &   &   &   &   &   &   &   & X & X & X &   \\
Ross Gore             & X & X & X & X & X &   &   & X &   &   &   &   &   \\
Michael Jarret        & X & X & X & X & X & X &   & X &  &   &   &   &   \\
William Kennedy       & X &   & X &   &   &   &   &   &   &   &   & X & X \\
Andrew Maciejunes     &   &   & X & X &   & X &   &   & X &   &   &   &   \\
Maxim Malikov      & X & X & X & X & X & X & X &   & X &   &   &   &   \\
Samuel Mendelson      & X &   &   &   &   &   &   &   &   &   &   &   &   \\
\bottomrule
\end{tabular}%
}
\caption{Author contributions per CRediT taxonomy (X indicates contribution). \label{tab:contributions}}
\end{table*}

\bibliographystyle{unsrt}  
\bibliography{references}  




%

\end{document}